\documentclass[11pt,a4paper]{article}
\usepackage[T1]{fontenc}
\usepackage{lmodern,microtype}

\usepackage[margin=27mm,headheight=14pt]{geometry}
\usepackage{amsmath,amssymb,amsthm,mathtools,bm,mathrsfs}
\usepackage{enumitem,booktabs,fancyhdr}
\usepackage[dvipsnames]{xcolor}
\usepackage[colorlinks=true,linkcolor=MidnightBlue,citecolor=MidnightBlue,urlcolor=MidnightBlue]{hyperref}
\hypersetup{pdftitle={BCOV quantization: flat space, relative CE theory, and Calabi--Yau manifolds},pdfauthor={Si Li and Minghao Wang},pdfsubject={Perturbative BCOV quantization on fixed Calabi--Yau manifolds and trivalent Feynman graphs for primary fields}}
\setlist{itemsep=3pt,topsep=5pt,leftmargin=1.7em}
\numberwithin{equation}{section}
\newtheorem{theorem}{Theorem}[section]
\newtheorem{proposition}[theorem]{Proposition}
\newtheorem{lemma}[theorem]{Lemma}

\theoremstyle{definition}
\theoremstyle{remark}\newtheorem*{remark}{Remark}
\newcommand{\C}{\mathbb C}\newcommand{\R}{\mathbb R}
\newcommand{\Z}{\mathbb Z}
\newcommand{\h}{\hbar}\newcommand{\dd}{\mathrm d}
\newcommand{\dc}{\mathrm d^c}\newcommand{\bdel}{\bar\partial}
\newcommand{\Tr}{\operatorname{Tr}}\newcommand{\PV}{\operatorname{PV}}
\newcommand{\Symb}{\operatorname{Symb}}\newcommand{\Aut}{\operatorname{Aut}}

\newcommand{\prim}{\mathrm{prim}}\newcommand{\Ev}{\operatorname{Ev}}
\newcommand{\Nor}{\operatorname{Nor}}\newcommand{\res}{\mathrm{res}}
\newcommand{\cM}{\mathcal M}\newcommand{\Mbar}{\overline{\mathcal M}}
\newcommand{\Mhat}{\widehat{\mathcal M}}\newcommand{\cE}{\mathcal E}
\newcommand{\cJ}{\mathcal J}

\newcommand{\cR}{\mathcal R}\newcommand{\cO}{\mathcal O}

\DeclareMathOperator*{\argmin}{arg\,min}

\begin{document}

 \title{\LARGE\bfseries On the exact quantization of BCOV Theory on Calabi-Yau Manifolds}

\author{
  Si Li \thanks{
    Yau Mathematical Sciences Center, Tsinghua University, Beijing, China.
  }
  \and
  Minghao Wang \thanks{
    Department of Mathematics, East China Normal University,
    Shanghai, China.\newline
    \hspace*{1.8em}Institute of Mathematical Sciences, the Chinese University of Hong Kong, Hong Kong, China.
    }
}

  \date{}

  \maketitle

\begin{abstract}

We announce the solution of constructing perturbative quantum BCOV invariants on arbitrary Calabi-Yau manifolds. In this note, we outline the construction and present an explicit solution. The quantization includes all descendant fields and all loop orders, and expresses the quantum BCOV interaction by an explicit integration over the moduli space of marked Riemann surfaces. This canonical solution provides a mathematical construction of all higher genus BCOV invariants on compact Calabi-Yau manifolds, which are conjecturally equivalent to all higher genus Gromov-Witten invariants on the mirror.   

\end{abstract}

\begingroup
\small
\setlength{\parskip}{0pt}
\setcounter{tocdepth}{2}
\pdfbookmark[1]{Contents}{contents}
\tableofcontents
\endgroup
\clearpage

\section*{Introduction}
\phantomsection
\addcontentsline{toc}{section}{Introduction}
\markright{Introduction}

The Kodaira-Spencer theory of gravity was introduced by Bershadsky--Cecotti--Ooguri--Vafa 
\cite{bcov} as the closed string field theory corresponding to the B-twisted
topological string. The original formulation of BCOV theory was on Calabi-Yau manifolds of dimension $3$. In \cite{costello-li}, Costello and the first author extended Kodaira-Spencer gravity to Calabi-Yau manifolds of arbitrary dimension by incorporating with descendant fields, and formulated the general quantization framework. This is a quantization in Costello's effecitve BV formalism \cite{costello}, and a solution of effective quantum master equation (which we call BCOV quantization) leads to a rigorous construction of higher genus B-model invariants on a Calabi-Yau manifold (see also \cite{li-renorm}). In \cite{costello-li}, a canonical BCOV quantization was shown to exist on elliptic curves by deformation theory. This was later explicitly constructed in \cite{li-VOA}, in which the full quantum BCOV invariants on elliptic curves were proven to be equivalent to all genus Gromov-Witten invariants on the mirror.  In \cite{costello-li-open-closed}, a BCOV quantization was shown to exist locally on the flat space $\mathbb C^N$ by coupling with holomorphic Chern-Simons theory and anomaly cancellation argument. At one loop, this anomaly cancellation is the topological string analogue of the Green–Schwarz mechanism \cite{costello-li-GS}.

In this note, we present an explicit quantization of BCOV theory on any compact
Calabi--Yau manifold in the Costello-Li framework \cite{costello-li}. The construction includes all descendants and is valid at every loop order. More precisely, let $X$ be a closed connected Calabi--Yau manifold of complex dimension $N$, with holomorphic volume form $\Omega$ and a
K\"ahler metric $g$ satisfying \eqref{eq:volume-metric}.
The field complex is $\PV(X)[[\zeta]][2]$, with differential
$Q=\bdel-\zeta\partial_\Omega$. Write a field 
$$
A(\zeta)=\sum_{k\geq0}\zeta^kA_k, \qquad A_k\in \PV(X)
$$
$A_0$ is the primary field and $A_k$ has descendant index $k$. We construct a local interaction \eqref{eq:curved-interaction}
$$
 \cJ_X= \sum_{g\geq 0} \hbar^g I_g. 
$$
$I_0$ is the genus-zero
classical descendant BCOV interaction of Costello--Li
\cite{costello-li}. The higher loop $I_g$'s are local functionals consisting of holomorphic derivatives with coupling constants given by  explicit integrations over the moduli space of marked Riemann surfaces (see below). 

At every scale $L>0$, let $I_X[L]$ be the effective functional that is built up from Feynman graph integrals with vertex $\cJ_X$ and propagator $P_0^L=\int_0^L \partial \bar\partial^* e^{-t[\bar\partial,\bar\partial^*]}dt$. Since BCOV theory is a holomorphic field theory, the above Feynman graph integrals are free of ultraviolet divergence and they are well-defined by Cauchy principal values \cite{wangyan,renorm}, without the need for singular ultraviolet counterterms. Theorem \ref{thm:curved} proves that the resulting interaction satisfies the renormalized BV quantum master equation
\[
 QI_X[L]+\tfrac12\{I_X[L],I_X[L]\}_{g,L}
                    +\h\Delta_{g,L}I_X[L]=0.
\]
The BV bracket $\{-,-\}_{g,L}$ and BV Laplacian $\Delta_{g,L}$ are defined by the scale-$L$ heat kernel.
We abbreviate this equation to QME. Theorem \ref{thm:primary} then shows that, for
$1\leq N\leq3$, amplitudes with primary external fields are computed by
the classical cubic interaction with the same propagator. Applying the general framework \cite{costello-li}, we obtain the higher genus BCOV invariants on $X$ from the partition function, i.e. from the infrared limit $I_X[\infty]$.

The paper is organized as follows.

Section \ref{sec:flat} constructs the quantized local interaction on $\C^N$ from Green
functions on marked curves. Stokes' theorem on their moduli space gives
the relation among vertices needed for the QME. Comparing its boundary
terms with the short-time limits of heat-kernel integrals, and applying
the convergence results of Wang \cite{wang}, proves the flat-space
quantization.

The exact formula of the interaction can be described as follows. In local Betti coordinates for conserved sources $P=(\mathbf{p}_i)$, the theta
form on the Jacobian is (see Lemma \ref{lem:betti})
\[
 \Theta_g(P)=\sum_{a=1}^N\sum_{\mu=1}^g
       \dd x_\mu^a\wedge\dd y_\mu^a.
\]
It determines the coefficients
\[
 V_{g;\bm k}(P)=\int_{\cM_{g,n}}
       \frac{\Theta_g(P)^{Ng}}{(Ng)!}\prod_{i=1}^n\psi_i^{k_i},
 \qquad |\bm k|=(3-N)g+n-3.
\]
Replacing each source $\mathbf{p}_i$ by a holomorphic derivative on the $i$th input
gives the local BCOV interaction
\[
 I(A)=\sum_{\substack{g,n,\bm k\geq0\\2g-2+n>0}}
       \frac{\h^g}{n!}\Tr\left.
       V_{g;\bm k}(\partial_{z_1},\ldots,\partial_{z_n})
       \prod_{i=1}^n A_{k_i}(z_i)\right|_{z_1=\cdots=z_n=z}
       =I_0+I_+.
\]
Here $A_{k_i}$ are the descendant components, $I_0$ is the genus-zero
classical descendant BCOV interaction of Costello--Li
\cite{costello-li}, and $I_+$ consists of the positive-genus terms. $\Tr$ is the trace map on polyvector fields, given by integration again the trivial holomorphic volume form. Thus the theta-form formula extends the classical interaction to
all genera. The construction and normalization are detailed in Section \ref{sec:flat}.

Section \ref{sec:curved} extends the exact formula of the interaction to a compact Calabi--Yau manifold $X$ of arbitrary complex dimension. Section \ref{sec:primary} proves the assertion about cubic vertices.

A detailed account of the construction and proofs of the results announced here will appear in a subsequent paper \cite{li-wang}. The above formula suggests an interpretation in terms of the B-twisted topological $\sigma$-model coupled to worldsheet gravity. We leave this for future investigation. 

All series are formal in genus, arity, and descendant degree.
Convergence refers to each coefficient integral. The total genus of a
connected graph is $G=b_1(\Gamma)+\sum_vg_v$, where $g_v$ is the genus of
the vertex at $v$. All products and commutators use the total grading.

\bigskip
\noindent\textbf{Acknowledgements.}
The authors would like to thank Zhengping Gui, Xun Lin, Nicolai Reshetikhin‌, Junrong Yan, and Shing-Tung Yau for helpful
discussions and valuable suggestions. S.~L. is supported by NSFC No.12571068. The authors are grateful to ChatGPT 6 Astra for valuable assistance with the exploration of some of the proof strategies and calculations. The authors take full responsibility for the mathematical content and exposition of this paper.

\section{Quantization on flat space}\label{sec:flat}

In this section, we construct the quantum interaction of BCOV theory
on $\C^N$ using Green's functions on marked Riemann surfaces.
The interaction coefficients are expressed as integrals over the
moduli spaces of these surfaces. We then prove that the resulting
effective interaction satisfies the quantum master equation.
The proof uses two applications of Stokes' theorem.
The first, on the moduli spaces of marked Riemann surfaces,
gives relations among the interaction vertices through boundary
degenerations. We identify these boundary contributions with the
short-time boundary terms of heat-kernel integrals.
The second application, on Schwinger parameter space,
then gives the quantum master equation.

\subsection{Fields and propagators}
We recall Costello--Li's descendant formulation \cite{costello-li} of BCOV theory, with
their descendant parameter $t$ replaced by $-\zeta$ for later convenience. Set
\[
 \PV^{i,j}(\C^N)=\Omega^{0,j}(\C^N,\wedge^iT^{1,0}),\qquad
 \cE=\PV(\C^N)[[\zeta]][2],\qquad Q_0=\bdel-\zeta D,
\]
where $D=\partial_\Omega$ is the divergence operator for $\Omega=\dd z^1\cdots\dd z^N$. We assign cohomology degree
$|\zeta|=2$, $|\h|=6-2N$. Write
$$
A(\zeta)=\sum_{k\geq0}\zeta^kA_k, \qquad A_k\in \PV(\C^N).
$$
$A_0$ is the primary field and
$A_k$ has descendant index $k$. An element of
$\zeta^k\PV^{i,j}[2]$ has cohomology degree $i+j+2k-2$.

Let $\theta_a$ be the polyvector generators and
$\upsilon_a=\dd\bar z^a$. The divergence, gauge-fixing operator, and
Laplacian are
\begin{equation}\label{eq:flat-operators}
 D=\sum_a\partial_{z^a}\partial_{\theta_a},\qquad
 q_0=-\sum_a\partial_{z^a}\partial_{\upsilon_a},\qquad
 H_0=[\bdel,q_0]=-\sum_a\partial_{z^a}\partial_{\bar z^a}.
\end{equation}
Odd derivatives are left derivatives. 

For compactly supported inputs $\mu\in \PV_c(\C^N)$, we fix the trace map
\begin{equation}\label{eq:trace}
 \Tr(\mu)=c_N\int_{\C^N}
[\theta_1\cdots\theta_N\upsilon_1\cdots\upsilon_N]\mu\,\dd^{2N}z,  \qquad c_N=\frac{(-1)^{N(N-1)/2}}{\pi^N}.
\end{equation}
Here the bracket $[-]$ extracts the coefficient of $\theta_1\cdots\theta_N\upsilon_1\cdots\upsilon_N$, or equivalently, performs Berezin integration with the corresponding ordering convention.  

Let $E_t$ be the kernel of the heat operator $e^{-tH_0}$ on polyvectors, including the inverse trace
tensor. Its scalar factor is $(\pi t)^{-N}e^{-|z-w|^2/t}$. Define
\begin{equation}\label{eq:flat-kernels}
 K_t=DE_t,\qquad R_t=Dq_0E_t,\qquad
 P_{\epsilon,L}=\int_\epsilon^L R_t\,\dd t,\qquad
 \Delta_t=\partial_{K_t}.
\end{equation}

Following the set-up in \cite{costello}, 
for a kernel $J$, let $\partial_J$ contract one unordered pair of primary
slots. Thus the effective propagator $P_{\epsilon,L}$ supplies internal edges, and $K_t$ defines
the regularized BV Laplacian and its derived BV bracket $\{-,-\}_t$. Vertices may still have arbitrary descendant inputs. The heat-kernel time $t$ in the propagator is called a
\emph{Schwinger parameter}. In a Feynman graph, each internal
edge $e$ carries its own parameter $t_e$. The product of their
ranges is the \emph{Schwinger parameter space}; with cutoffs
$\epsilon$ and $L$, each $t_e$ ranges over $[\epsilon,L]$.

\begin{lemma}[Heat-kernel homotopy \cite{costello}]\label{lem:flat-homotopy}
With heat-time forms placed before field variables,
\begin{equation}\label{eq:kernel-homotopy}
 [Q_0,\partial_{P_{\epsilon,L}}]=\Delta_\epsilon-\Delta_L,
 \qquad
 ([Q_0,-]+\dd_t)(\Delta_t-\dd t\,\partial_{R_t})=0.
\end{equation}
\end{lemma}

For a local interaction $J$, the connected renormalization group flow is
defined by \cite{costello,costello-li}
\begin{equation}\label{eq:RG}
 e^{W(P,J)/\h}=e^{\h\partial_P}e^{J/\h}.
\end{equation}
Its connected graphs carry the usual factor $1/|\Aut(\Gamma)|$ and
contain at least one original interaction vertex. A vertex of genus $g_v$
contributes $\h^{g_v}$, so the total power is
$\h^{b_1(\Gamma)+\sum_vg_v}$. Graphs with no external legs are included.

\subsection{Green forms and the local interaction}
Let $(C;z_1,\ldots,z_n)\in\cM_{g,n}$, with $2g-2+n>0$.
Choose smooth Hermitian metrics $h_i$ on the cotangent lines $L_i$ over
$\Mbar_{g,n}$ \cite{knudsen}, using the conventions
\[
 \dc=\frac{i}{2\pi}(\bdel-\partial),\qquad
 \psi_i=-\dd\dc\log\|\ell_i\|_{h_i}.
\]
Here $\ell_i$ is a local nonzero holomorphic frame.
The operators $\dd,\dc$ differentiate the curve and its marked points. On the other hand, we use $\dd_C$ for the de Rham differential on a single curve $C$. These differentials do not apply to the sources introduced below. 

For a conserved source, we mean a set of vectors $P=(\mathbf{p}_i)$ in $\R^N$ satisfying $\sum_i \mathbf{p}_i=0$. Let $u_P$ be the vector-valued
Green potential on the curve $C$, unique modulo constants, satisfying
\[
 \dd_C*\dd_Cu_P=2\pi\sum_i \mathbf{p}_i\delta_{z_i},\qquad \sum_i \mathbf{p}_i=0.
\]
In a coordinate $w_i$ centered at $z_i$, write
$u_P=\mathbf{p}_i\log|w_i|+u_i^{\mathrm{reg}}+o(1)$ and put
$a_i=\|\dd w_i|_{z_i}\|_{h_i}$. Define
\begin{equation}\label{eq:finite-energy}
 F_h(P)=\frac12\sum_i \mathbf{p}_i\cdot u_i^{\mathrm{reg}}
             +\frac12\sum_i|\mathbf{p}_i|^2\log a_i,
\end{equation}
and
\begin{equation}\label{eq:theta-definition}
 L_P=\frac12\sum_i|\mathbf{p}_i|^2\psi_i,\qquad
 \beta=\dc F_h,\qquad \Theta_g(P)=L_P+\dd\beta.
\end{equation}
The Euclidean pairing is extended bilinearly when the sources are
complexified.

The theta form $\Theta_g(P)$ serves as a basic building block
in our construction of the quantum interaction.
We give another useful characterization of this form in terms of the periods of the Green current. Let
\[
 j_P=\frac{1}{2\pi}*\dd_Cu_P
\]
be the current associated with the Green potential.
Choose locally a symplectic homology basis
$(A_\mu,B_\mu)_{\mu=1}^g$, represented by cycles avoiding
the marked points and transported continuously as the curve varies.
Define
\[
 x_\mu^a=\int_{B_\mu}j_P^a,\qquad
 y_\mu^a=-\int_{A_\mu}j_P^a,
 \qquad 1\leq a\leq N,\quad 1\leq\mu\leq g,
\]
where $j_P^a$ denotes the $a$th component of $j_P$.
\begin{lemma}[Green forms]\label{lem:betti}
The function $F_h$ is well-defined, independent of the coordinates $w_i$
and the additive constant of $u_P$, and quadratic in $P$.
The form $\Theta_g$ is independent of the metrics $h_i$. In terms of periods $(x_\mu^a,y_\mu^a)$,
\begin{equation}\label{eq:betti}
 \Theta_g(P)=\sum_{a=1}^N\sum_{\mu=1}^g
                 \dd x_\mu^a\wedge\dd y_\mu^a,\qquad
 \dd\Theta_g=0,\qquad \Theta_g^{Ng+1}=0.
\end{equation}
 More generally, if the sources span
a space of dimension $\rho$, then $\Theta_g^{\rho g+1}=0$ (rank identity).
\end{lemma}
\begin{remark}[Abel--Jacobi interpretation]
When $\mathbf{p}_i\in\mathbb Z^N$, each divisor
$\sum_i \mathbf{p}_i^a z_i$ has degree zero, and
$(x^a,y^a)$ modulo integers gives the Betti coordinates
of its usual Abel--Jacobi image in $\operatorname{Jac}(C)$.
Under the resulting map to the family of products
$\operatorname{Jac}(C)^N$, the form $\Theta_g(P)$ is the
pullback of the canonical closed two-form whose restriction
to each fiber represents the product principal polarization.
For general real sources, we use the period formulas above directly.
\end{remark}
The Green and period conventions agree with \cite{fay}.  The closedness and rank identity in Lemma \ref{lem:betti} will
be used in the boundary argument later.

For $\bm k=(k_1,\ldots,k_n)$ with $k_i\geq0$, define
\begin{equation}\label{eq:vertices}
 V_{g;\bm k}(P)=\int_{\cM_{g,n}}
       \frac{\Theta_g(P)^{Ng}}{(Ng)!}\prod_i\psi_i^{k_i},
 \qquad |\bm k|=(3-N)g+n-3.
\end{equation}
The integral is an orbifold integral. Unstable data and incorrect degrees
give zero. The degree condition $|\bm k|=(3-N)g+n-3$ follows from $\dim_\C\cM_{g,n}=3g-3+n$.

Replace each $\mathbf{p}_i$ by a holomorphic derivative on the $i$th input to set
\begin{equation}\label{eq:flat-interaction}
 I(A)=\sum_{\substack{g,n,\bm k\geq0\\2g-2+n>0}}
 \frac{\h^g}{n!}\Tr\left.
 V_{g;\bm k}(\partial_{z_1},\ldots,\partial_{z_n})
 \prod_iA_{k_i}(z_i)\right|_{z_i=z}=I_0+I_+.
\end{equation}
Here $I_+$ includes the powers of $\h$ in positive genus. Since the
coefficients are defined on $\sum_i \mathbf{p}_i=0$, a polynomial representative
is chosen before substitution. Different representatives differ by a
total holomorphic derivative and give the same integrated local functional.
The convergence and polynomiality needed for this definition are stated
in Lemma \ref{lem:vertex-properties} below.

\begin{lemma}[Classical coefficients]\label{lem:classical}
At genus zero,
\begin{equation}\label{eq:classical}
 V_{0;\bm k}=\frac{(n-3)!}{\prod\limits_i k_i!},\qquad
 I_0=\sum_{n\geq3}\frac1{n!}
 \sum_{|\bm k|=n-3}\frac{(n-3)!}{\prod\limits_i k_i!}
                    \Tr(A_{k_1}\cdots A_{k_n}).
\end{equation}
Thus $I_0$ is the classical descendant BCOV interaction, and its
restriction to primary fields is $\Tr(A_0^3)/6$.
\end{lemma}
The coefficient formula is the genus-zero  intersection formula 
$
\int_{\cM_{0,n}}\prod\limits_{i=1}^n\psi_i^{k_i}=\binom{n-3}{k_1,\cdots, k_3}.
$


\subsection{Extension and boundary values of the Green forms}
We record the two facts about degeneration needed below: smooth extension
of the forms and their restrictions to the boundary.

Let $\Mhat_{g,n}$ be the logarithmic real resolution of
$\Mbar_{g,n}$ obtained as follows. In a plumbing coordinate, write
\[
 q_e=e^{-T_e-2\pi i\tau_e},\qquad
 \rho_e=T_e^{-1},\qquad \tau_e\in\R/\Z.
\]
Retain $\tau_e$ at $\rho_e=0$ and resolve the coordinate corners by
real blow-ups, beginning with the deepest corners. Thus the boundary
records the relative rates of the neck lengths as well as the node
angles. This is the conductance version of the resolution in
\cite[Definition 1.10 and Lemma 1.13]{grushevsky-krichever-norton}.
The construction is understood in orbifold charts.

For a stable graph $H$ with at least one edge, let
$\partial_H \subset \widehat{\mathcal M}_{g,n}$
denote the corresponding boundary face. Write $E=|E(H)|$, $h=b_1(H)$,
and
\[
 X_H=\prod_v\cM_{g_v,n_v},\qquad g=h+\sum_vg_v,\qquad
 \lambda_e=\frac{T_e}{\sum_fT_f}.
\]
The face is a node-angle torus bundle over
$X_H\times\Delta_H^{\res}$, modulo graph automorphisms.
Here $\Delta_H^{\res}$ is the projective length simplex with its relative
rates resolved, and its interior is $\lambda_e>0$, $\sum_e\lambda_e=1$. Corners describe successive degenerations.

Orient the edges and let $P_v$ be the sum of the original sources at $v$.
The graph current $j^0=(j_e^0)$ is the unique minimizer
\begin{equation}\label{eq:graph-current}
 j^0=\argmin_j\frac12\sum_e\lambda_e|j_e|^2,
 \qquad
 \sum_{e\text{ outgoing at }v}j_e-
 \sum_{e\text{ incoming at }v}j_e=P_v.
\end{equation}
At the outgoing and incoming ends of an edge, assign node sources
$-j_e^0$ and $+j_e^0$, respectively. Together with the original sources
they form a conserved collection $P'_v$ on each component. The two
half-edges of a self-loop are treated separately. Set $f_e=|j_e^0|^2/2$.

\begin{lemma}[Smooth extension]\label{lem:green-extension}
The forms $\beta$, $\Theta_g$, and all their source derivatives extend
smoothly to $\Mhat_{g,n}$. The currents $j^0$ and the functions $f_e$
are smooth on $\Delta_H^{\res}$. These extensions are compatible on
overlapping plumbing charts and along every nested boundary stratum.
\end{lemma}

\begin{proposition}[Boundary values]\label{prop:energy-degeneration}
In a local trivialization of the node-angle torus, normalize the
cotangent metrics at the node markings by $\|\dd w_\alpha\|=1$.
The boundary restrictions are
\begin{equation}\label{eq:boundary-forms}
 \begin{aligned}
 \Theta_g|_{\partial_H}
    &=\sum_v\Theta_{g_v}\bigl(P'_v(j^0)\bigr)
             +\dd_\lambda b-\sum_e\dd f_e\wedge\dd\tau_e,\\
 \beta|_{\partial_H}
    &=b-\sum_e f_e\,\dd\tau_e,
 \end{aligned}
\end{equation}
where
\begin{equation}\label{eq:boundary-basic}
 b=\sum_v\beta_{g_v}\bigl(P'_v(j^0)\bigr).
\end{equation}
Here $\dd_\lambda$ differentiates the relative lengths.
In particular, at fixed $\lambda$ and node angles,
\begin{equation}\label{eq:boundary-components}
 \Theta_g|_{TX_H}=\sum_v\Theta_{g_v}\bigl(P'_v(j^0)\bigr).
\end{equation}
These formulas may be differentiated in conserved source directions and
applied successively at corners.
\end{proposition}

\begin{lemma}\label{lem:vertex-properties}
The integrals \eqref{eq:vertices} converge absolutely and are homogeneous
polynomials of degree $2Ng$ in the conserved sources. They are independent
of the chosen cotangent-line metrics.
\end{lemma}

\subsection{Stokes' theorem on curve moduli}
We next show that the rank identity in Lemma \ref{lem:betti} gives a differential identity. This allows us to apply Stokes' theorem  and use
Proposition \ref{prop:energy-degeneration} to evaluate the boundary integral.

For $\sum_i\eta_i=0$, define
\[
 D_\eta=\sum_{i,a}\eta_i^a\partial_{p_i^a},\qquad
 \Lambda_\eta=\sum_i(\mathbf{p}_i\cdot\eta_i)\psi_i.
\]
We first use commuting $\eta_i$. All identities are linear in $\eta$,
so the odd source symbols of the field theory can be substituted afterwards.

\begin{samepage}
\begin{lemma}[Transgression of the rank identity]\label{lem:transgression}
For $|\bm k|=(3-N)g+n-4$,
\begin{equation}\label{eq:boundary-form}
 \dd\left((D_\eta\beta)\frac{\Theta_g^{Ng}}{(Ng)!}
                      \prod_i\psi_i^{k_i}\right)
 =-\Lambda_\eta\frac{\Theta_g^{Ng}}{(Ng)!}\prod_i\psi_i^{k_i}.
\end{equation}
\end{lemma}
\begin{proof}
Differentiate $\Theta_g^{Ng+1}=0$ in the direction $\eta$ and use
$D_\eta\Theta_g=\Lambda_\eta+\dd(D_\eta\beta)$.
The forms $\Theta_g$ and $\psi_i$ are closed, and the lemma follows.
\end{proof}
\end{samepage}

\pagebreak[3]
Order the edges $0,\ldots,E-1$ and orient the projective simplex by
$\dd\lambda_1\cdots\dd\lambda_{E-1}$. Define the boundary flux
\begin{equation}\label{eq:flux}
 \Phi_H=(-1)^{E-1}\sum_{a=0}^{E-1}(-1)^a(D_\eta f_a)
 \dd f_0\wedge\cdots\wedge\widehat{\dd f_a}\wedge\cdots
                         \wedge\dd f_{E-1},
\end{equation}
where $\dd f_e$ in this formula differentiates only the projective lengths.
The current in \eqref{eq:graph-current} and the flux \eqref{eq:flux}
are the only graph data needed for the comparison with heat kernels.

\begin{proposition}[Boundary identity]\label{prop:boundary-identity}
For $|\bm k|=(3-N)g+n-4$,
\begin{equation}\label{eq:boundary-coefficients}
 \sum_i(\mathbf{p}_i\cdot\eta_i)V_{g;\bm k+e_i}(P)
 =\sum_{\substack{[H,(g_v)]\\E(H)=Nb_1(H)+1}}
 \frac1{|\Aut(H,(g_v))|}
 \int_{\Delta_H^{\res}}\Phi_H
           \prod_vV_{g_v;(\bm k_v,0_{\mathrm{int}})}(P'_v).
\end{equation}
The graphs are connected and stable, $g=b_1(H)+\sum_vg_v$, and the
external legs are labeled. Each node marking has descendant index zero.
\end{proposition}
\begin{proof}[Proof sketch]
By Lemma \ref{lem:green-extension}, Stokes' theorem applies to
\eqref{eq:boundary-form} and gives
\[
 -\int_{\Mhat_{g,n}}\Lambda_\eta
       \frac{\Theta_g^{Ng}}{(Ng)!}\prod_i\psi_i^{k_i}
 =\int_{\partial\Mhat_{g,n}}(D_\eta\beta)
       \frac{\Theta_g^{Ng}}{(Ng)!}\prod_i\psi_i^{k_i}.
\]
Use Proposition \ref{prop:energy-degeneration} to evaluate the right-hand
side. Integration over the node angles selects $E=Nh+1$ and removes
$D_\eta b$. With the boundary orientation, the surviving flux is \eqref{eq:flux}. The component integrals give the
product of vertex coefficients, and the orbifold gluing degree gives
the automorphism factor in \eqref{eq:boundary-coefficients}.
\end{proof}

\subsection{Comparison with the heat-kernel integrals}
The short-time collapse of a heat-kernel graph is governed by the same
minimizing current $j^0$ and flux $\Phi_H$ as the degeneration of a
marked curve. The comparison below identifies its local boundary
contribution with the corresponding term in \eqref{eq:boundary-coefficients}.

Fix a connected graph $H$ with ordered edges and polynomial holomorphic
vertex symbols $F_v$. Following  \cite{wang}, we place the extended kernel
\[
 \widehat K_{t_e}=K_{t_e}-\dd t_e\,R_{t_e}
\]
on every edge and apply the vertex operators. Contract using the trace
pairing and integrate all vertex positions except one output point.
Denote the heat-time degree $E-1$ part by $\omega_H(t;F_v)$: each term
has one $K$ edge and $E-1$ propagator terms, with signs determined by
the ordered product.

For a fixed upper cutoff $L$, we use the compactified Schwinger parameter space (see \cite[Section~3.1]{wang} and \cite{renorm}). Its added boundary
faces record the relative rates at which edge parameters tend
to zero. Near the boundary face where all edge parameters vanish,
write $t_e=\rho\lambda_e$, with $\sum_e\lambda_e=1$, and let
$\iota_H$ denote the face $\rho=0$.
The local short-time boundary operator is
\begin{equation}\label{eq:heat-boundary-definition}
 \cO_H(F_v)=\int_{\Delta_H^{\res}}\iota_H^*\omega_H(t;F_v).
\end{equation}
This is a local operator before taking the output trace or dividing by
the graph automorphism factor.
The existence and locality of the boundary restriction follows from 
\cite[Theorems 5--6 and Propositions 9, 11]{wang}.

The reduced symbol replaces derivatives on the inputs by their even and
odd momenta, subject to total momentum conservation. The odd source
direction is $\eta$, as in \eqref{eq:flux}.

\begin{proposition}[Heat-kernel and curve-boundary comparison]
\label{prop:two-pushforwards}
The operator \eqref{eq:heat-boundary-definition} vanishes unless
$E=Nh+1$. At this edge number its reduced symbol is
\begin{equation}\label{eq:heat-operation}
 \Symb^{\mathrm{red}}\cO_H(F_v)
   =\int_{\Delta_H^{\res}}\Phi_H\prod_vF_v(P'_v).
\end{equation}
Here $P'_v$ is determined by the same minimizing current
\eqref{eq:graph-current} as in the curve degeneration. The integral converges on the resolved simplex.
The identity also holds for subgraph insertions and successive collapses.
\end{proposition}

\subsection{Stokes' theorem and the QME}
The first Stokes argument and the heat-kernel comparison give local
cancellation. The second Stokes argument, on the heat-time parameters,
turns this cancellation into the QME at every positive scale.

\begin{proposition}[Local cancellation]\label{prop:flat-master}
For the interaction \eqref{eq:flat-interaction}, $Q_0I$ is cancelled by
the sum of the local short-time boundary contributions
$\cO_H$, with their genus powers and graph automorphism factors.
\end{proposition}
\begin{proof}[Proof sketch]
The integrated $\bdel$ part of a holomorphic vertex is zero, while
$-\zeta D$ has reduced symbol
\[
 \Symb^{\mathrm{red}}(Q_0I)_{g,\bm k}
       =-\sum_i(\mathbf{p}_i\cdot\eta_i)V_{g;\bm k+e_i}.
\]
Proposition \ref{prop:boundary-identity} expresses its negative as a sum
of curve-boundary terms, and Proposition \ref{prop:two-pushforwards}
identifies each term with the corresponding heat-kernel boundary
operator. Integration by parts turns equality of reduced symbols into
equality of integrated local functionals. The insertion statement follows
from the same comparison; at zero arity, the source flux vanishes.
\end{proof}

\begin{theorem}[Flat quantization]\label{thm:flat}
The coefficients \eqref{eq:vertices} define a local interaction on the
full descendant field complex. Its effective interaction exists
coefficientwise and satisfies
\begin{equation}\label{eq:flat-QME}
 I[L]=W(P_{0,L},I),\qquad
 Q_0I[L]+\tfrac12\{I[L],I[L]\}_L+\h\Delta_LI[L]=0,
 \quad L>0.
\end{equation}
The statement includes the vacuum coefficient.
\end{theorem}
\begin{proof}[Proof sketch]
The vertices are well-defined by Lemma \ref{lem:vertex-properties}, and
the ultraviolet cutoff can be removed coefficientwise by \cite{wang}.
The homotopy identity in Lemma \ref{lem:flat-homotopy} converts the action
of $Q_0$ on the internal edges into a heat-time differential. Apply
Stokes' theorem to the heat-time parameters of each connected graph.
At the short-time boundary, collapsing subgraphs cancel the action of
$Q_0$ on the vertices by Proposition \ref{prop:flat-master}. At scale
$L$, cutting an edge gives the BV bracket if the graph disconnects and
the BV Laplacian otherwise. These boundary terms give precisely
\eqref{eq:flat-QME}.
\end{proof}

\section{Quantization on a Calabi--Yau manifold}\label{sec:curved}\label{sec:CE}

Let $X$ be a closed connected Calabi--Yau manifold of complex dimension
$N$, with holomorphic volume form $\Omega$ and a fixed
K\"ahler metric $g$ satisfying
\begin{equation}\label{eq:volume-metric}
 \frac{\omega^N}{N!}=c\,\Omega\wedge\overline{\Omega},
 \qquad c\text{ constant}.
\end{equation}
Condition \eqref{eq:volume-metric} makes $g$ Ricci-flat, hence real analytic
in holomorphic coordinates \cite[Theorem 6.1]{deturck-kazdan}.
Normalize so that $\det(g_{i\bar j})=1$ in holomorphic volume
coordinates. The Chern connection $\nabla$ is torsion-free and preserves
$\Omega$. The field complex and differential are
\begin{equation}\label{eq:curved-fields}
 \cE_X=\PV(X)[[\zeta]][2],\qquad
 \PV^{i,j}(X)=\Omega^{0,j}(X,\wedge^iT^{1,0}X),\qquad
 Q=\bdel-\zeta D,\quad D=\partial_\Omega.
\end{equation}
The trace is the global version of \eqref{eq:trace} in volume charts. See \cite{costello-li} for further details. 

To define the interaction, let $\mathfrak g$ be the Lie algebra
of formal volume-preserving vector fields on
$\widehat{\mathbb C}^N$ with its standard volume form.
We use the Chevalley--Eilenberg (CE) complex of $\mathfrak g$;
see \cite[Chapter~1, \S3]{fuks}.
The universal CE background $B$ is the tautological
$\mathfrak g$-valued one-cochain, defined by
\[
 B(v)=v,\qquad v\in\mathfrak g.
\]
We regard its values as polyvector fields of degree one,
with the odd CE coefficients placed before the odd polyvector
generators. Thus $B$ is even in the total grading and satisfies
$DB=0$. Similarly, we write $A_k$ for the tautological field coordinate
of descendant degree $k$: for a field
$\alpha=\sum_{k\geq0}\zeta^k\alpha_k$, we have
\[
 A_k(\alpha)=\alpha_k.
\]
Thus $A(\zeta)=\sum_{k\geq0}\zeta^kA_k$ denotes the universal field. From the flat interaction $I=I_0+I_+$, define

\begin{equation}\label{eq:RB}
 \cR_BA_k=\delta_{k0}B-BA_{k+1},\qquad \cR_BB=0,
 \qquad e^{\cR_B}A(\zeta)=B+[e^{-B/\zeta}A(\zeta)]_{\geq0}.
\end{equation}
Here, products such as $BA_{k+1}$ are taken using the wedge product of
polyvector fields, $[\,\cdot\,]_{\geq0}$ retains the nonnegative powers of $\zeta$.
Prolonging $\cR_B$ to holomorphic jets gives the equivariant interaction
\begin{equation}\label{eq:lift}
 \widehat I_B=I_0+e^{\cR_B}I_+.
\end{equation}
Its positive-genus part has no term without ordinary inputs. Choose
its representative $\Nor$ by eliminating an ordinary output momentum
and averaging over the output choices. For $n$ ordinary inputs of
descendant degrees $\bm\ell$, use weights $\ell_r/|\bm\ell|$ when
$|\bm\ell|>0$, and $1/n$ otherwise. In this representative, every
background input carries at least two holomorphic derivatives.

With the one-form coefficient placed first, set
\begin{equation}\label{eq:curvature}
 \mathcal R_{ij}=-(\bdel\Gamma^k_{ij})\theta_k.
\end{equation}
Replace field derivatives by covariant derivatives and background jets
by curvature jets according to
\begin{equation}\label{eq:evaluation}
 \begin{aligned}
 \partial_{i_1}\cdots\partial_{i_r}A_k
   &\longmapsto\nabla_{i_1}\cdots\nabla_{i_r}A_k,\\
 \partial_{i_1}\cdots\partial_{i_r}B
   &\longmapsto\nabla_{i_1}\cdots\nabla_{i_{r-2}}
                        \mathcal R_{i_{r-1}i_r},\qquad r\geq2.
 \end{aligned}
\end{equation}
Write $\Ev_{\nabla,\mathcal R}$ for this replacement. Covariant
derivatives act on all tensor indices; contractions use the parallel
volume tensors, and graded multiplication precedes the trace.
The resulting interaction on $X$ is
\begin{equation}\label{eq:curved-interaction}
 \boxed{\quad
 \cJ_X=I_0+\Ev_{\nabla,\mathcal R}\Nor(e^{\cR_B}I_+).
 \quad}
\end{equation}

\begin{theorem}[Quantization on a Calabi--Yau manifold]\label{thm:curved}
The interaction \eqref{eq:curved-interaction} defines a perturbative
quantization of BCOV theory on $X$: at every scale $L>0$, its effective
interaction $I_X[L]$ satisfies the quantum master equation on the full
descendant field complex, at all loop orders and including the vacuum
coefficient. The ultraviolet cutoff in its heat-kernel Feynman expansion
can be removed coefficientwise by the convergence results of
Wang--Yan \cite{wangyan,renorm}.
\end{theorem}

\begin{proposition}[Metric independence]\label{prop:metric-independence}
Fix the complex structure of $X$ and the volume form $\Omega$.
For $a\in H^\bullet(\PV(X),\bdel)[[\zeta]][2]$, let $a_g$ denote
its $g$-harmonic representative. The limit
\begin{equation}\label{eq:metric-independent-F}
 I_X[\infty](a)=\lim_{L\to\infty}I_X[L](a_g)
\end{equation}
exists coefficientwise and is independent of the choice of K\"ahler metric
$g$ satisfying \eqref{eq:volume-metric}. This includes the vacuum term $I_X[\infty](0)$.
\end{proposition}

\section{The primary sector for \texorpdfstring{$\dim \leq3$}{1 <= N <= 3}}\label{sec:primary}

Fix $(X,\Omega,g)$ as in Section \ref{sec:curved}. Let
$\cE_{\prim}=\PV(X)[2]$ be embedded at descendant degree zero, and put
\begin{equation}\label{eq:cubic}
 I_{\mathrm{cub}}(A_0)=\frac16\Tr(A_0^3).
\end{equation}

\begin{lemma}[Primary vertex reduction]\label{lem:primary-vertices}
For $N=1,2,3$, write
$\cJ_X=I_0+\sum_{g\geq1}\h^g\cJ_{X,g}$. Then
\begin{equation}\label{eq:primary-vertices}
 \cJ_{X,g}|_{\cE_{\prim}}=0\quad(g\geq1),\qquad
 \cJ_X|_{\cE_{\prim}}=I_{\mathrm{cub}}.
\end{equation}
This holds on all primary fields, without imposing $DA_0=0$.
\end{lemma}
\begin{proof}
On primary fields, \eqref{eq:RB} reduces to $A_0\mapsto A_0+B$.
For $n$ ordinary and $m$ background inputs, \eqref{eq:vertices} gives
\[
 n+m=3-(3-N)g=N-(3-N)(g-1)\leq N\qquad(g\geq1),
\]
contradicting the rank requirement $n+m\geq N+1$ from
Lemma \ref{lem:betti}.
Curvature evaluation preserves these conditions. In genus zero,
\eqref{eq:classical} with all descendants zero forces $n=3$ and gives
the coefficient $1/6$.
\end{proof}

\begin{theorem}\label{thm:primary}
For $N=1,2,3$, only classical cubic vertices contribute to connected
Feynman graphs with primary external inputs, at every loop order and
including vacuum graphs.
\end{theorem}

\phantomsection

\addcontentsline{toc}{section}{References}
{\small
\begin{thebibliography}{99}
\setlength{\itemsep}{3pt}
\setlength{\parskip}{0pt}

\bibitem{arbarello-cornalba}
E.~Arbarello and M.~Cornalba,
\emph{Combinatorial and algebro-geometric cohomology classes on the moduli
spaces of curves}, J.~Algebraic Geom.~\textbf{5} (1996), 705--749.
\href{https://arxiv.org/abs/alg-geom/9406008}{arXiv:alg-geom/9406008}.

\bibitem{bcov}
M.~Bershadsky, S.~Cecotti, H.~Ooguri and C.~Vafa,
\emph{Kodaira--Spencer theory of gravity and exact results for quantum
string amplitudes}, Comm.~Math.~Phys.~\textbf{165} (1994), 311--428.
\href{https://arxiv.org/abs/hep-th/9309140}{arXiv:hep-th/9309140}.

\bibitem{costello}
K.~Costello, \emph{Renormalization and Effective Field Theory},
Mathematical Surveys and Monographs, vol.~170,
American Mathematical Society, Providence, RI, 2011.
\href{https://doi.org/10.1090/surv/170}{doi:10.1090/surv/170}.

\bibitem{costello-li}
K.~Costello and S.~Li,
\emph{Quantum BCOV theory on Calabi--Yau manifolds and the higher genus
B-model}, \href{https://arxiv.org/abs/1201.4501}{arXiv:1201.4501} (2012).

\bibitem{costello-li-open-closed}
K.~Costello and S.~Li,
\emph{Quantization of open-closed BCOV theory, I}, \href{https://arxiv.org/abs/1505.06703}{arXiv:1505.06703 } (2015).

\bibitem{costello-li-GS}
K.~Costello and S.~Li,
\emph{Anomaly cancellation in the topological string}, Adv. Theor. Math. Phys. Vol. 24, No. 7 (2020), pp. 1723-1771. \href{https://dx.doi.org/10.4310/ATMP.2020.v24.n7.a2}{doi:10.4310/ATMP.2020.v24.n7.a2}.

\bibitem{deturck-kazdan}
D.~M.~DeTurck and J.~L.~Kazdan,
\emph{Some regularity theorems in Riemannian geometry},
Ann.~Sci.~\'{E}cole Norm.~Sup.~(4) \textbf{14} (1981), no.~3, 249--260.
\href{https://doi.org/10.24033/asens.1405}{doi:10.24033/asens.1405}.

\bibitem{fay}
J.~D.~Fay, \emph{Theta Functions on Riemann Surfaces},
Lecture Notes in Mathematics, vol.~352, Springer, Berlin, 1973.
\href{https://doi.org/10.1007/BFb0060090}{doi:10.1007/BFb0060090}.

\bibitem{fuks}
D.~B.~Fuks, \emph{Cohomology of Infinite-Dimensional Lie Algebras},
Monographs in Contemporary Mathematics, Consultants Bureau, New York, 1986.
\href{https://doi.org/10.1007/978-1-4684-8765-7}{doi:10.1007/978-1-4684-8765-7}.

\bibitem{grushevsky-krichever-norton}
S.~Grushevsky, I.~Krichever and C.~Norton,
\emph{Real-normalized differentials: limits on stable curves}, Russian Math. Surveys 74 (2019), no. 2, 265–324. \href{https://arxiv.org/abs/1703.07806v2}{arXiv:1703.07806}, v2 (2019). 

\bibitem{knudsen}
F.~F.~Knudsen, \emph{The projectivity of the moduli space of stable
curves, II: The stacks $M_{g,n}$}, Math.~Scand.~\textbf{52} (1983), 161--199.
\href{https://doi.org/10.7146/math.scand.a-12001}{doi:10.7146/math.scand.a-12001}.

\bibitem{li-renorm}
S.~Li, \emph{Renormalization method and mirror symmetry},
SIGMA \textbf{8} (2012), 101, 17~pp.
\href{https://doi.org/10.3842/SIGMA.2012.101}{doi:\allowbreak10.3842/SIGMA.2012.101}.

\bibitem{li-VOA}
S.~Li, \emph{Vertex algebras and quantum master equation}, J.Diff.Geom. \textbf{123} (2023), no 3, 461--521.
\href{https://doi.org/10.4310/jdg/1683307007}{doi.org/10.4310/jdg/1683307007}

\bibitem{li-index}
S.~Li, \emph{Quantization and Algebraic Index}, Surv. Differ. Geom. \textbf{28} (2023), no 1.
\href{https://dx.doi.org/10.4310/SDG.260101000010}{doi.org/10.4310/SDG.260101000010}

\bibitem{li-wang}
S.~Li and M.~Wang, \emph{Quantum BCOV invariants on Calabi-Yau manifolds}, in preparation. 

\bibitem{wang}
M.~Wang, \emph{Feynman graph integrals on $\C^d$}, Commun. Math. Phys. 406.5 (2025): 116
\href{https://link.springer.com/article/10.1007/s00220-025-05298-2}{DOI: 10.1007/s00220-025-05298-2}.

\bibitem{wangyan}
M.~Wang and J.~Yan, \emph{Feynman graph integrals on K\"ahler manifolds},
\href{https://arxiv.org/abs/2507.09170v5}{arXiv:2507.09170}, v5.

\bibitem{renorm}
M.~Wang and J.~Yan, \emph{Renormalizations in holomorphic field theories on
K\"ahler manifolds},
\href{https://arxiv.org/abs/2608.00546v2}{arXiv:2608.00546}, v2.
\end{thebibliography}
}
\end{document}